\documentclass{article}
\usepackage{qic,epsfig}

\usepackage[utf8]{inputenc}
\usepackage[english]{babel}
\usepackage[T1]{fontenc}
\usepackage{amsmath,amssymb,amsthm,nicefrac,hyperref,braket,qcircuit,array}
\newcolumntype{?}{!{\vrule width 1.2pt}}
\usepackage{tikz}
\usepackage{algorithm,algpseudocode,qtree}

\usetikzlibrary{shapes, arrows, positioning}
\usetikzlibrary {arrows.meta}
\usetikzlibrary {bending}

\newtheorem{lmm}{Lemma}

\newtheorem{cor}{Corollary}

\DeclareMathOperator{\conv}{\operatorname{conv}}

\begin{document}
\setlength{\textheight}{8.0truein}

\runninghead{The signaling dimension  of two-dimensional and
  polytopic  systems}{Shuriku Kai and Michele Dall'Arno}

\normalsize\textlineskip
\thispagestyle{empty}


\vspace*{0.88truein}

\alphfootnote


\centerline{\bf THE  SIGNALING DIMENSION  OF TWO-DIMENSIONAL
  AND POLYTOPIC SYSTEMS}

\vspace*{0.37truein}

\centerline{\footnotesize SHURIKU KAI}
\vspace*{0.015truein}
\centerline{\footnotesize\it Department  of Computer Science
  and  Engineering,  Toyohashi   University  of  Technology,
  Japan}

\vspace*{10pt}

\centerline{\footnotesize MICHELE DALL'ARNO}
\vspace*{0.015truein}
\centerline{\footnotesize\it Department  of Computer Science
  and  Engineering,  Toyohashi   University  of  Technology,
  Japan}
\baselineskip=10pt
\centerline{\footnotesize\it michele.dallarno.mv@tut.jp}

\vspace*{0.225truein}


\vspace*{0.21truein}

\abstracts{The  signaling dimension  of  any given  physical
  system represents its classical  simulation cost, that is,
  the  minimum dimension  of a  classical system  capable of
  reproducing all the input/output correlations of the given
  system.   The  signaling  dimension  landscape  is  vastly
  unexplored; the  only non-trivial systems  whose signaling
  dimension is  known -- other  than quantum systems  -- are
  the    octahedron    and    the   composition    of    two
  squares.}{Building  on  previous   results  by  Matsumoto,
  Kimura, and Frenkel, our first result consists of deriving
  bounds  on the  signaling  dimension of  any  system as  a
  function of  its Minkowski  measure of asymmetry.   We use
  such bounds to  prove that the signaling  dimension of any
  two-dimensional system (i.e.   with two-dimensional set of
  admissible states, such as polygons and the real qubit) is
  two if and only if such  a set is centrally symmetric, and
  three otherwise, thus conclusively settling the problem of
  the signaling dimension for such systems.}  {Guided by the
  relevance of  symmetries in  the two dimensional  case, we
  propose a branch and bound division-free algorithm for the
  exact computation of the symmetries of any given polytope,
  in  polynomial  time in  the  number  of vertices  and  in
  factorial time in the dimension  of the space.  Our second
  result  then consist  of  providing an  algorithm for  the
  exact computation of the  signaling dimension of any given
  system, that outperforms  previous proposals by exploiting
  the   aforementioned  bounds   to   improve  its   pruning
  techniques   and  incorporating   as   a  subroutine   the
  aforementioned symmetries-finding algorithm.  We apply our
  algorithm  to compute  the  exact value  of the  signaling
  dimension  for  all  rational Platonic,  Archimedean,  and
  Catalan  solids, and  for  the  class of  hyper-octahedral
  systems up to dimension five.}

\vspace*{10pt}


\vspace*{1pt}\textlineskip

\section{Introduction}

Generalized  probabilistic theories~\cite{CDP11}  (GPTs), of
which quantum theory is an example, provide the most general
model to  describe how  correlations between  observed input
and output events  should be computed.  In  this sense, they
all represent generalizations of  classical theory. But what
is the classical cost of simulating  any given such a GPT by
means of classical  theory?  The answer to  this question is
formally given  by the  \textit{signaling dimension}  (for a
recent  overview,  see   Ref.~\cite{Dal22}),  that  is,  the
dimension  of   the  smallest  classical  system   that  can
reproduce all  the input/output correlations that  the given
system is capable of.

As  such, the  signaling  dimension of  any given  classical
system   is,  by   definition,   equal   to  the   dimension
(e.g. maximum number of perfectly distinguishable states) of
such a system. However, even  for systems as familiar as the
quantum ones,  quantifying the signaling dimension  has been
elusive  for decades.   During all  this time,  it has  been
arguably  part of  the quantum  folklore that  the signaling
dimension must be equal to the Hilbert space dimension (once
again,  the  maximum  number  of  perfectly  distinguishable
states).    Finally,   in    2015,   in   a   groundbreaking
work~\cite{FW15}  Frenkel  and Weiner  conclusively  settled
this debate by proving the  correctness of this belief; they
did  so   with  an   elegant  proof  that   leverages  graph
theoretical results.

Frenkel  and Weiner's  work fueled  further research  on the
topic.   Following their  work, it  was shown~\cite{DBTBV17}
that there exists GPTs  whose systems, while consistent with
classical  and quantum  theory  at the  level of  space-like
correlations,  exhibit   anomalies  --  quantifiable   as  a
superadditivity  of  the  signaling dimension  under  system
composition  known as  \textit{hypersignaling}  -- in  their
time-like  correlations.   In  other  words,  the  signaling
dimension plays for space-like correlations a role analog to
that  played by  the no-signaling  principle for  space-like
correlations  (hence the  name).   That work  also lied  the
foundations  for the  characterization  of  the polytope  of
input/output classical  correlations and the  computation of
the signaling dimension of any given system.

Following   this    line   of   research,    Matsumoto   and
Kimura~\cite{MK18}  proved  a  very  interesting  connection
between the signaling dimension and the Minkowski measure of
asymmetry.   Building   on  that  and  on   the  results  of
Ref.~\cite{DBTBV17},  and  applying again  graph-theoretical
results  as  in  his previous  result,  Frenkel~\cite{Fre22}
recently  provided bounds  on the  signaling dimension  as a
function of the  dimension of the linear span  of the states
of  the system  only.  In  the  same work,  the author  also
computed the  signaling dimension  of a simple  system whose
set of admissible states forms an octahedron.

At  the  same   time,  Doolittle  and  Chitambar~\cite{DC21}
extended   the  results   of   Ref.~\cite{DBTBV17}  to   the
characterization    of    the    polytope    of    classical
correlations. While they cleverly  exploit the symmetries of
the polytope to simplify  the characterization of its facets
(the  vertices are  trivial  to characterize,  but what  one
typically  needs  in  computations are  the  facets),  their
method   still  requires   the  complete   and  simultaneous
enumeration  of   all  the  vertices,  whose   number  grows
exponentially  in the  number of  states of  the system  and
factorially in the dimension of  the space.  That makes such
an approach impractical in most cases.

Recently, it was shown~\cite{DTB24}  that such an exhaustive
enumeration is actually unnecessary  if one is interested in
computing  the  signaling   dimension.   By  exploiting  the
symmetries~\cite{GMCD10} (if  any) of the set  of admissible
states and the structure of the correlation matrices induced
on  such a  set  by  the extremal  measurements~\cite{Par99,
  DLP05, Dav78} of the system,  an algorithm was devised and
implemented  that is  capable of  computing \textit{exactly}
the  signaling  dimension  of  systems  as  complex  as  the
(hypersignaling)  composition  of   two  square  systems  --
typically      introduced~\cite{DT10}      to      reproduce
Popescu-Rohrlich correlations~\cite{PR94} --  in a matter of
minutes.   This   result,  along  with   the  aforementioned
bounds~\cite{Fre22} derived by Frenkel, represents the state
of the art of the signaling dimension problem.

Here,  we progress  in two  directions.  First,  building on
Ref.~\cite{Fre22}, we provide upper  and lower bounds on the
signaling dimension of any given  system. We show that, when
specified to two-dimensional systems (any system whose state
space is two  dimensional, such as polygonal  systems or the
real  qubit), our  results characterize  in closed-form  the
signaling  dimension.    Specifically,  we  show   that  the
signaling dimension  is two  for any such  a system  that is
centrally  symmetric;  and  it  is  three  otherwise.   This
settles   the  issue   of   the   signaling  dimension   for
two-dimensional systems.

The second direction  we progress in is the  case of systems
whose set of admissible states is a polytope.  We provide an
exact,  branch  and   bound,  division-free  algorithm  that
exactly characterizes  the symmetries of any  given polytope
in  polynomial  time  in  the  number  of  vertices  and  in
factorial time  in the dimension. Our  pruning rule provides
an heuristic  speedup (while  preserving the  correctness of
the result)  over previous proposals~\cite{DBK23}.   We show
how to exploit such a  result, along with the aforementioned
bounds, to reduce the complexity of the exact computation of
the  signaling  dimension, and  we  apply  our algorithm  to
compute the exact expression  of the signaling dimension for
certain   rational  Archimedean   and  Catalan   solids  and
hyper-octahedra,  thus generalizing  the result  obtained by
Frenkel~\cite{Fre22} for the (three-dimensional) octahedron.

The paper is structured as follows. In Section~\ref{sect:2d}
we formalize the signaling dimension, we introduce bounds as
a function of the asymmetry,  and we conclusively settle the
problem  of   the  signaling  dimension  for   systems  with
two-dimensional  set of  admissible states.   Guided by  the
relevance of the symmetries for the two-dimensional case, in
Section~\ref{sect:symmetries}  we provide  an exact,  branch
and  bound, division-free  symmetries-finding algorithm  for
polytopes in arbitrary dimension. In Section~\ref{sect:algo}
we address the problem  of computing the signaling dimension
for systems whose set of admissible states is a polytope, by
providing an  algorithm that, exploiting  the aforementioned
bounds and  symmetries-finding subroutine,  exactly computes
the  signaling dimension  of any  given polytopic  system in
finite  time.  In  Section~\ref{sect:applications} we  apply
our  algorithm to  the  exact computation  of the  signaling
dimension for  systems whose set  of admissible states  is a
rational    regular   or    quasi-regular   solid    or   an
hyper-octahedron. We summarize our  results and discuss some
open problems in Section~\ref{sect:conclusion}.

\section{Two dimensional systems}
\label{sect:2d}

In this section,  building upon previous results~\cite{MK18,
  Fre22} by Matsumoto, Kimura, and Frenkel, we provide upper
and lower  bounds on  the signaling  dimension of  any given
system as a function of its central symmetry.

For  any  given system,  whose  (convex)  set of  admissible
states we denote with  $\mathcal{S} \in \mathbb{R}^\ell$, we
denote with $\operatorname{lin.dim}(  \mathcal{S} ) := \ell$
and $\operatorname{aff.dim}( \mathcal{S} )  := \ell - 1$ the
linear dimension and the affine dimension, respectively.

Let us  denote with $\mathcal{P}_\mathcal{S}^{m \to  n}$ the
polytope of $m$-input/$n$-output  correlations attainable by
system $\mathcal{S}$, that is
\begin{align*}
  \mathcal{P}_{\mathcal{S}}^{m \to n} :=  \left\{ p \; \Big|
  \;  \exists  \left\{ \omega_i  \right\}_{i=1}^m  \subseteq
  \mathcal{S},   \left\{   e_j  \right\}_{j=1}^n   \subseteq
  \mathcal{E} \textrm{  s.t. } p_{i,j} =  \omega_i \cdot e_j
  \right\},
\end{align*}
where $\mathcal{E}$  denotes the set of  effects (typically,
but  not   necessarily,  the  set  of   linear  non-negative
functionals    over   $\mathcal{S}$).    We   denote    with
$\operatorname{sig.dim}(\mathcal{S})$      the     signaling
dimension of $\mathcal{S}$, given by
\begin{align}
  \label{eq:sigdim}
  \operatorname{sig.dim}   \left(  \mathcal{S}   \right)  :=
  \inf_{\substack{d                                      \in
      \mathbb{N}\\\mathcal{P}_{\mathcal{S}}^{m     \to    n}
      \subseteq \mathcal{P}_d^{m  \to n}}} d,
\end{align}
for any $m$  and $n$.  It was  proven in Ref.~\cite{DBTBV17}
that  the   signaling  dimension   (for  any   system  whose
$\mathcal{S}$  is  not  trivially  a point)  is  bounded  as
follows:
\begin{align}
  \label{eq:bounds}
  2 \le \operatorname{sig.dim}  \left( \mathcal{S} \right)
  & \le \operatorname{lin.dim} \left( \mathcal{S} \right).
\end{align}

Finally, we denote  with $\operatorname{asymm} ( \mathcal{S}
)$ the Minkowski measure of asymmetry~\cite{Gru63}, that is,
the smallest  dilation factor  needed to cover  the mirrored
set  $- \mathcal{S}  :=  \{  - \omega  \;  |  \; \omega  \in
\mathcal{S} \}$ up to a translation, that is
\begin{align}
  \label{eq:asymm}
  \operatorname{asymm}   \left(   \mathcal{S}   \right)   :=
  \inf_{\substack{\lambda >  0\\ c  \in \mathbb{R}^{\ell}\\-
      \left(  \mathcal{S}  -  c  \right)  \subseteq  \lambda
      \left( \mathcal{S} - c \right)}} \lambda.
\end{align}
It immediately follows from this definition that
\begin{align}
  \label{eq:bounds2}
  1 \le \operatorname{asymm}  \left( \mathcal{S} \right) \le
  \operatorname{aff.dim} \left( \mathcal{S} \right),
\end{align}
where   the   first   inequality    is   tight   (that   is,
$\operatorname{asymm}(\mathcal{S})  =  1$)  if and  only  if
$\mathcal{S}$  is   centrally  symmetric,  and   the  second
inequality         is         tight        (that         is,
$\operatorname{asymm}(\mathcal{S})  = \operatorname{aff.dim}
( \mathcal{S} )$) if and only if $\mathcal{S}$ is a simplex.

Notice  the  formal  analogy   between  the  definitions  of
signaling dimension and  asymmetry in Eqs.~\eqref{eq:sigdim}
and~\eqref{eq:asymm}.     Both   quantities    represent   a
``shrinking  parameter''  for  a set,  each  quantity  being
defined  as  the  value  of  such a  parameter  at  which  a
particular  set   inclusion  relation  is   satisfied.   The
difference   between  the   two  quantities   lies  in   the
definitions  of such  sets:  in the  case  of the  signaling
dimension, the set is the polytope $\mathcal{P}_d^{m \to n}$
of  input/output  correlations  attainable  by  a  classical
system of dimension $d$; in the case of the asymmetry, it is
the set of admissible states itself.

Notice  also  the  formal  analogy  between  the  bounds  in
Eqs.~\eqref{eq:bounds}   and~\eqref{eq:bounds2}.    However,
while necessary and sufficient  conditions for the tightness
of the  bounds in Eq.~\eqref{eq:bounds2} are  known, neither
necessary  nor sufficient  conditions for  the tightness  of
Eq.~\eqref{eq:bounds}  are   known.   The   following  lemma
addresses this gap by providing necessary conditions.

\begin{lmm}
  \label{lmm:bounds}
  For  any  system  with  set  $\mathcal{S}$  of  admissible
  states, the signaling dimension satisfies
  \begin{align*}
    \operatorname{sig.dim} \left( \mathcal{S}  \right) & \le
    \operatorname{aff.dim}   \left(    \mathcal{S}   \right)
    \textrm{ if $\mathcal{S}$ is centrally symmetric},
  \end{align*}
  and
  \begin{align*}
    \operatorname{sig.dim} \left( \mathcal{S}  \right) & \ge
    3 \textrm{ if $\mathcal{S}$ is not centrally symmetric}.
  \end{align*}
  or     equivalently     the    first     inequality     in
  Eq.~\eqref{eq:bounds}  is strict  if $\mathcal{S}$  is not
  centrally symmetric and the second inequality is strict if
  $\mathcal{S}$ is centrally symmetric.
\end{lmm}

\begin{proof}
  The   statement  simply   follows   by  putting   together
  Eq.~\eqref{eq:bounds2} with results from Refs.~\cite{MK18}
  and~\cite{Fre22}.
  
  Let  us  first consider  the  case  when $\mathcal{S}$  is
  centrally symmetric.  In  this case $\operatorname{asym} (
  \mathcal{S} ) = 1$  due to Eq.~\eqref{eq:bounds2}.  Due to
  Thm.~1 of Ref~\cite{MK18} one has $\operatorname{inf.stor}
  ( \mathcal{S})  = \operatorname{asymm} (\mathcal{S})  + 1$
  (we omit  here the  definition of the  information storage
  $\operatorname{inf.stor}(\mathcal{S})$;    the    previous
  equation can be taken as  its definition as a function of
  the   asymmetry   $\operatorname{asymm}   (\mathcal{S})$).
  Hence,    for    centrally    symmetric    S    one    has
  $\operatorname{inf.stor}  (\mathcal{S})  =   2$.   Due  to
  Thm.~1.2,    point   (2),    of   Ref.~\cite{Fre22},    if
  $\operatorname{inf.stor}     (    \mathcal{S}     )    \le
  \operatorname{aff.dim}     (    \mathcal{S}     )$    then
  $\operatorname{sig.dim}     (     \mathcal{S}    )     \le
  \operatorname{aff.dim}  (  \mathcal{S}   )$.   Hence,  for
  centrally     symmetric      $\mathcal{S}$     one     has
  $\operatorname{sig.dim}     (     \mathcal{S}    )     \le
  \operatorname{aff.dim} ( \mathcal{S} )$.  Hence, the first
  part of the statement follows.

  Let us  then consider the  case when $\mathcal{S}$  is not
  centrally  symmetric.   Then,  $\operatorname{inf.stor}  (
  \mathcal{S}  )  >  2$ due  to  Eq.~\eqref{eq:bounds2}  and
  Thm.~1 of Ref.~\cite{MK18}.  Then, $\operatorname{sig.dim}
  (  \mathcal{S} )  >  2$ due  to $\operatorname{sig.dim}  (
  \mathcal{S}    )$     being    an    upper     bound    to
  $\operatorname{inf.stor}    (    \mathcal{S}    )$    (see
  Ref.~\cite{MK18}    or    Thm.~1.2,    point    (1),    of
  Ref.~\cite{Fre22}).   Hence,  the   second  part   of  the
  statement follows.
\end{proof}

Notice  that, while  central symmetry  is necessary  for the
signaling   dimension   to    saturate   its   lower   bound
$\operatorname{sig.dim}(\mathcal{S}) = 2$,  such a condition
is not sufficient: for instance, Frenkel showed~\cite{Fre22}
that  the  signaling  dimension   of  the  octahedral  (thus
centrally  symmetric) system  is three.   Analogously, while
central asymmetry  is necessary for the  signaling dimension
to         saturate          its         upper         bound
$\operatorname{sig.dim}(\mathcal{S})                       =
\operatorname{lin.dim}(\mathcal{S})$,  such  a condition  is
not    sufficient:    for     instance,    it    has    been
shown~\cite{DBTBV17}  that the  signaling  dimension of  the
composition of  two square  systems (which is  not centrally
symmetric) is five, while its linear dimension is nine.

Despite  these   considerations,  Lemma~\ref{lmm:bounds}  is
strong enough  to pin  down the  signaling dimension  of any
two-dimensional  system (systems  whose affine  dimension is
two, as is the case e.g. for polygonal theories and the real
qubit) as a  function of its geometry only, as  shown by the
following corollary.

\begin{cor}
  \label{cor:sigdim2d}
  Any system $\mathcal{S}$ such that $\operatorname{aff.dim}
  ( \mathcal{S} ) = 2$ has signaling dimension given by
  \begin{align}
    \label{eq:sigdim2d}
    \operatorname{sig.dim} \left( \mathcal{S} \right) =
    \begin{cases}
      2 & \textrm{if  $\mathcal{S}$ is centrally symmetric},
      \\3 & \textrm{otherwise}.
    \end{cases}
  \end{align}
\end{cor}

\begin{proof}
  The statement immediately follows  by a direct application
  of Lemma~\ref{lmm:bounds} and of Eq.~\eqref{eq:bounds}.
\end{proof}

For  instance,  for  regular   polygonal  systems  with  $m$
vertices (states), the signaling dimension  is two if $m$ is
even and three otherwise.

\section{An exact symmetries-finding algorithm}
\label{sect:symmetries}

Motivated  by  the role  played  by  the symmetries  in  the
computation  of the  signaling dimension  of two-dimensional
systems, as shown  in the previous section,  in this section
we  address the  problem of  finding the  symmetries of  any
given point  set. How  this problem  fits within  the bigger
problem  of  computing  the   signaling  dimension  will  be
illustrated in the next section.  A trivial variation of our
approach also provides  a way to test the  congruence of two
given  point   sets;  that  is,  whether   there  exists  an
orthogonal transformation mapping the one into the other.

Previous   literature~\cite{Hig85,    WWV85,   BK02,   KR16}
approached the  symmetry-finding problem from  the geometric
viewpoint,    that   is,    by   looking    for   orthogonal
transformations   (linear   transformations   whose   matrix
representation   is   orthogonal,   that  is,   angle-   and
length-preserving)  that act  as permutations  of the  point
set.  That  is to say,  such algorithms depend on  the field
structure of their  input.  To this aim, they  assume a real
computational model, that is, an unphysical machine that can
exactly  store  any  real  number and  can  exactly  perform
arithmetic, trigonometric, and other functions over reals in
finite time.

In Ref.~\cite{DBK23},  in the  unrelated context  of quantum
guesswork,  the   symmetries-finding  problem   was  instead
approached  from the  combinatorial viewpoint,  that is,  by
looking for permutations  of the labels of the  set that act
as orthogonal transformations.  By using established results
on  Gram  matrices,  it  was possible  to  avoid  explicitly
dealing  with orthogonal  transformations altogether.   That
way, a symmetries-finding algorithm  was presented that only
depends on the weaker ring structure (that is, the operation
of  division  is  not  assumed).   That  approach  therefore
corresponds to  an integer  computational model  that solely
assumes the ability to store  integer numbers and to perform
additions and  multiplication in finite time  (any potential
buffer overflow can  be detected and the  computation can be
restarted  by allocating  additional memory),  thus allowing
for closed-form analytical solutions on physical machines.

The number  of permutations grows factorially  in the number
of points, hence the complexity  of any algorithm based on a
naive   exhaustive  search   is   factorial.   However,   by
exploiting a  well-known rigidity property of  simplices, in
Ref.~\cite{DBK23}  it was  also shown  that without  loss of
generality  it   suffices  to   search  over  a   subset  of
permutations  whose  size  only grows  polynomially  in  the
number  of   points,  although  still  factorially   in  the
dimension of the space.

To amend this factorial scaling in the dimension, we propose
here a  branch and bound  algorithm that pushes  forward the
ideas  of Ref.~\cite{DBK23}  by first  reframing the  search
over such  permutations as  the exploration  of a  tree, and
then ``pruning''  those branches that can  provably be shown
not  to  lead  to  any symmetry.   Our  algorithm  therefore
provides   a   heuristic  speedup   over   Ref.~\cite{DBK23}
(however, we recall  that the final result  is guaranteed to
be exact) and, generally, a  better scaling in the dimension
of  the space,  although  the worst  case  scaling is  still
factorial.

A   symmetry   can   be   represented   as   an   orthogonal
transformation $O$ (that is, a length- and angles-preserving
transformation) that acts as a permutations on the labels of
the extremal vectors in set $\mathcal{S}$. In formula
\begin{align}
  \label{eq:sym}
  O \omega_i = \omega_{\sigma \left( i \right)},
\end{align}
for any $i$, where $\sigma$ denotes a permutation of indexes
$i$'s.

Equation~\eqref{eq:sym}   shows  that   symmetries  can   be
represented  as  permutations  of  the labels  that  act  as
orthogonal   matrices;   the  one-to-one   mapping   between
orthogonal   transformations  and   label  permutations   is
guaranteed by  $\mathcal{S}$ being  a spanning set.   If one
also recalls that two labeling  are related by an orthogonal
matrix if  and only if  the corresponding Gram  matrices are
identical, one  immediately has an equivalent  condition for
permutation $\sigma$ to be a symmetry, that is
\begin{align}
  \label{eq:sym2}
  \omega_i \cdot \omega_j = \omega_{\sigma \left( i \right)}
  \cdot \omega_{\sigma \left( j \right)}
\end{align}
for  any  $i$  and  $j$.    Notice  that  the  condition  in
Eq.~\eqref{eq:sym2} is  purely combinatorial, as  opposed to
the geometric condition in Eq.~\eqref{eq:sym}.

Such  was   the  combinatorial  approach  adopted   to  find
symmetries in  Ref.~\cite{DBK23}.  The algorithm  therein i)
iterated        over        all       permutations        of
$\operatorname{aff.dim}(\mathcal{S})$ vectors through Heap's
or     Johnson-Trotter's      algorithms;     ii)     tested
Eq.~\ref{eq:sym2}  for  such  a permutation;  finally,  iii)
extended the test to all the remaining vectors in polynomial
time.  Overall, the algorithm's complexity was polynomial in
$m$   (the    number   of   vectors)   and    factorial   in
$\operatorname{aff.dim}(\mathcal{S})$.

Here,  we  discuss  a  different,  branch-and-bound  way  of
generating  all  permutations.  The generation  of  all  the
permutations  $\sigma$  can   be  achieved,  with  factorial
complexity, by recursively visiting the nodes of a tree; for
instance, when $m = 3$ the tree could be as follows:
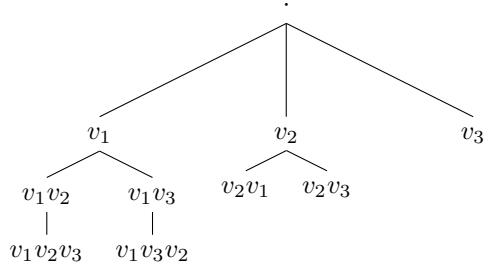
\begin{align*}
  \Tree[
    .{$\cdot$}
    [.{$v_1$}
      [.{$v_1 v_2   $}
        [.{$v_1   v_2 v_3  $} ]
      ]
      [.{$v_1 v_3   $}
        [.{$v_1   v_3 v_2  $} ]
      ]
    ]
    [.{$v_2$}
      [.{$v_2 v_1 $} ]
      [.{$v_2 v_3 $} ]
    ]
    [.{$v_3$} ]
  ]
\end{align*}

At the first step, the algorithm sets the first point in the
permutation, it  compares the  (so far,  $1 \times  1$) Gram
matrix  with  $G(v_1)$;  if  they  coincide,  the  algorithm
recursively calls  itself, chooses  the second point  in the
permutation, and compares the (now $2 \times 2$) Gram matrix
with $G  (v_1, v_2)$. The  algorithm proceeds this way  in a
depth-first exploration of the tree. Once it gets to a leaf,
the  same procedure  as  in  Ref.~\cite{DTB24} is  followed.
However,  with respect  to Ref.~\cite{DTB24},  the algorithm
does  \textit{not} necessarily  reach each  leaf, since  the
failure of the aforementioned comparison of Gram matrices on
the  ancestor  of a  leaf  proves  that  the leaf  does  not
represent  a symmetry  without the  need to  reach it.   For
instace,  in the  tree above,  such a  comparison failed  at
nodes $(v_2 v_1)$, $(v_2  v_3)$, and $(v_3)$. This procedure
is described in Algorithm~\ref{algo:symfact}.

\begin{algorithm}[h!]
  \caption{Symmetries finding (branch and bound)}
  \label{algo:symfact}
  \begin{algorithmic}
  \Require{$(v_1, \dots v_m)$ with $G (v_1, \dots v_d)$ invertible and $(v_{d+1}, \dots, v_m) = \operatorname{order} (v_1, \dots v_d)$}
    \Function{Node}{$\mathbf{v}$}
    \If{$G(v_1, \dots v_{|\mathbf{v}|}) \neq G(\mathbf{v})$}
    \Return
    \EndIf
    \If{$| \mathbf{v} | < d$}
    \For{$v \not\in \mathbf{v}$}
    \State \Call{Node}{$\operatorname{concat} ( \mathbf{v}, (v) )$}
    \EndFor
    \Else
    \State $\mathbf{u} \leftarrow \operatorname{order} ( \mathbf{v} )$
    \State $\mathbf{t} \leftarrow \operatorname{concat} ( \mathbf{v}, \mathbf{u} )$
    \If{$G ( v_1, \dots v_m) \neq G( \mathbf{t} )$} \Return
    \EndIf
    \State $\mathcal{S} \leftarrow \mathcal{S} \cup \{ \mathbf{t} \}$
    \EndIf
    \EndFunction
    
    \Call{Node}{$\cdot$}
  \end{algorithmic}
\end{algorithm}

\section{Polytopic systems}
\label{sect:algo}

In this section, we address the problem of computing exactly
the  signaling dimension  of  any given  system whose  state
space is a polytope. We  stress that our algorithm is exact,
that is, if the system  in input is represented exactly (for
instance, if its set $\mathcal{S}$ of admissible states is a
rational   polytope),   then   the  algorithm   outputs   in
\textit{finite}  time   the  \textit{exact}  value   of  its
signaling dimension.

We  assume  the  sets  $\mathcal{S}$  and  $\mathcal{E}$  of
admissible states and effects are given. If only one of them
is given, the other can  be found, by applying the so-called
no-restriction  hypothesis, using  the  techniques based  on
linear programming described  in Ref.~\cite{DTB24}.  We also
assume  that the  extremal measurements  are known.  If not,
they can be found using the techniques, also based on linear
programming, described in the same reference.

Once the symmetries  (if any) of the system  have been found,
the  algorithm  selects  without   loss  of  generality  one
representative measurement for  each equivalence class under
such  symmetries. For  each representative  measurement, the
algorithm tests  the signaling  dimension within  the bounds
provided  in Lemma~\ref{lmm:bounds}.  This,  along with  the
more effective way  to compute the symmetries,  allows for a
speedup with  respect to  the previously known  algorithm of
Ref.~\cite{DTB24}.

The  test itself  proceeds as  in Ref.~\cite{DTB24},  and we
summarize it here for  completeness.  First, the correlation
matrix between the states of the systems and the measurement
elements  is  computed,  producing   a  matrix  $p_{i,j}  :=
\omega_i \cdot e_j$.  Without loss of generality, the convex
hull $\conv ( \{ p_{i, \cdot} \}  )_i$ of the rows of such a
matrix is  considered; reducing the size  of the correlation
matrix  speeds up  later  steps of  the  algorithm. All  the
classical $m$-input/$n$-output correlation  matrices $\{ A_k
\}_k$ attainable by a  $d$-dimensional classical systems and
whose entries  are null wherever the  corresponding entry of
the correlation  matrix are null are  generated. As observed
in Ref.~\cite{DTB24},  this allows to compute  the signaling
dimension without the need for a simultaneous enumeration of
all the vertices  of $\mathcal{P}_d^{m \to n}$,  which is in
contrast,    for   instance,    with    the   proposal    of
Ref.~\cite{DC21}. To conclude the test, system $\mathcal{S}$
can be  simulated by  a classical $d$-dimensional  system if
and only if the correlation  matrix can be convex decomposed
in terms of the $A_k$'s, a fact that can be verified through
linear   programming.    The   algorithm  is   depicted   in
Fig.~\ref{fig:flowchart},  and  an   implementation  in  the
Python programming language is  provided as free software in
Ref.~\cite{KD24}.

\begin{figure}[h!]
  \begin{center}
    \begin{tikzpicture}
      \tikzset{terminal/.style={rounded rectangle, draw, text centered, text width=4cm, minimum height=3mm}};
      \tikzset{block/.style={rectangle, draw, text centered, text width=4cm, minimum height=3mm}};
      \tikzset{decision/.style={diamond, draw, text centered, text width=4cm, minimum height=3mm, aspect=3}};
      \tikzset{decision_big/.style={diamond, draw, text centered, text width=5cm, minimum height=3mm, aspect=3}};
      \node[terminal](start)at (0, 0){start};
      \node[block, below=0.5 of start](state){States $\mathcal{S}$};
      \node[block, below=0.5 of state](effect){Effects $\mathcal{E}$};
      \node[decision_big, below=0.5 of effect](extmeas){Another representative of ext.meas under symmetries?};
      \node[block, below=0.5 of extmeas](P){$p_{i,j} = \omega_i \cdot e_j$};
      \node[block, below=0.5 of P](dbound){$d \gets \begin{cases}
        2 & \textrm{if CS}\\
        3 & \textrm{otherwise}
      \end{cases}$};
      \node[decision_big, below=0.5 of dbound](check){Can $p$ be simulated by a classical system of dimension $d$?};
      \node[block, below=0.5 of check](dpp){$d \gets d + 1$};
      \node[decision_big, below=0.5 of dpp](dcheck){$d < \begin{cases}
        \operatorname{aff.dim} & \textrm{if CS} \\
        \operatorname{lin.dim} & \textrm{otherwise}
      \end{cases}$};
      \node[block, below=0.5 of dcheck](sigdim){$\operatorname{sig.dim} ( \mathcal{S} ) = \max(d)$};
      \node[terminal, below=0.5 of sigdim](end){end};
      \draw[arrows={-Latex[length=2mm]}](start)--(state);
      \draw[arrows={-Latex[length=2mm]}](state.south)node[right, yshift=-2.5mm]{rational double description method}--(effect);
      \draw[arrows={-Latex[length=2mm]}](effect.south)node[right, yshift=-2.5mm]{rational double description method}--(extmeas);
      \draw[arrows={-Latex[length=2mm]}](extmeas.south)node[right, yshift=-2.5mm]{Yes}--(P);
      \draw[arrows={-Latex[length=2mm]}](P)--(dbound);
      \draw[arrows={-Latex[length=2mm]}](dbound)--(check);
      \draw[arrows={-Latex[length=2mm]}](check.south) node[right, yshift=-2.5mm]{No} -- (dpp);
      \draw[arrows={-Latex[length=2mm]}](dpp.south)--(dcheck.north);
      \draw (dcheck.west) ++ (-0.5, 0) node[below, xshift=0.25cm]{Yes} -- (dcheck.west);
      \draw (dcheck.east) ++ (0.5, 0) node[below, xshift=-0.25cm]{No} -- (dcheck.east);
      \draw[arrows={-Latex[length=2mm]}](dcheck.east) ++ (0.5, 0) |- (extmeas.east);
      \draw[arrows={-Latex[length=2mm]}](dcheck.west) ++ (-0.5, 0) |- (check.west);
      \draw (check.east) ++ (0.99, 0) node[below, xshift=-0.5cm]{Yes} -- (check.east);
      \draw[arrows={-Latex[length=2mm]}](extmeas.west) ++ (-1.5, 0) |-(sigdim);
      \draw (extmeas.west) ++ (-1.5, 0) node[below, xshift=1cm]{No} -- (extmeas.west);
      \draw[arrows={-Latex[length=2mm]}](sigdim)--(end);
    \end{tikzpicture}
  \end{center}
  \fcaption{Flowchart  representation of  the algorithm  for
    the  exact   computation  of  the   signaling  dimension
    $\operatorname{sig.dim}(\mathcal{S})$  of any  given set
    $\mathcal{S}$     of     admissible     states     using
    Lemma~\ref{lmm:bounds}.}
  \label{fig:flowchart}
\end{figure}

The  problem  complexity  is  dominated  by  the  number  of
vertices of the polytope $\mathcal{P}_d^{m \to n}$, that is,
the  number   of  extremal   classical  $m$-input/$n$-output
correlation  matrices   $\{  A_k   \}_k$  attainable   by  a
$d$-dimensional  classical  systems.    Upon  denoting  with
$\binom{n}{k}$ the binomial coefficient  and with ${m \brace
  k}$  the   Stirling  number   of  the  second   kind,  one
has~\cite{DBTBV17}  that  the  number  $V$  of  vertices  of
$\mathcal{P}^{m \to n}_d$ is given by
\begin{align*}
  V = \sum_{k=1}^{d} k! \binom{n}{k} {m\brace k}.
\end{align*}
Such a  quantity grows  exponentially in  the number  $m$ of
states and factorially  in the number $n$  of effects; hence
the problem quickly becomes practically unfeasible as either
quantity  increases. 

\section{Applications}
\label{sect:applications}

The  algorithm has  been  applied to  explore the  signaling
dimension  landscape  of  GPTs whose  set  $\mathcal{S}$  of
admissible  states  is  a  rational  polytope,  that  is,  a
polytope  for  which  there  exists a  basis  in  which  the
coordinates of  each vertex are rational  numbers (integers,
up  to  a rescaling).   This  is  not  a limitation  of  the
algorithm  per se,  rather  a limitation  of the  underlying
linear    programming    and    double-description    method
subroutines,  and has  been  introduced in  order to  obtain
exact results.

Specifically, the sets that have been considered include the
rational  Platonic   solids  (octahedron,   whose  signaling
dimension had already been computed by Frenkel~\cite{Fre22},
and cube,  whose signaling dimension is  trivially two), the
rational   Archimedean    solids   (truncated   tetrahedron,
cuboctahedron, and  truncated octahedron), and  the rational
Catalan solids  (triakis tetrahedron,  rhombic dodecahedron,
and tetrakis hexahedron), all of which have affine dimension
equal to three (see Table~\ref{tab:regular}).

\begin{table}[h!]
  \begin{center}
    \begin{tabular}{|c|c|c|c|c|c|c|c|}
      \hline
      \textbf{$\mathcal{S}$} & $m$ & $\operatorname{aff.dim} ( \mathcal{S} )$ & CS & $|\mathcal{G}|$ & $|\mathcal{M}|$ & $|\mathcal{M}'|$ & $\operatorname{sig.dim} ( \mathcal{S} )$\\
      \hline
      Octahedron & 6 & 3 & True & 48 & 6 & 2 & 3\\
      Cube & 8 & 3 & True & 48 & 3 & 1 & 2\\
      Truncated tetrahedron & 12 & 3 & False & 24 & 6 & 3 & 3\\
      Triakis tetrahedron & 8 & 3 & False & 24 & 93 & 6 & 3\\
      Cuboctahedron & 12 & 3 & True & 48 & 41 & 6 & 3\\
      Rhombic dodecahedron & 14 & 3 & True & 48 & 20 & 3 & 2\\
      Truncated octahedron & 24 & 3 & True & 48 & 41 & 6 & 2\\
      Tetrakis hexahedron & 14 & 3 & True & 48 & 828 & 26 & 3\\
      \hline
    \end{tabular}
  \end{center}
  \tcaption{Exact value of the  signaling dimension for GPTs
    as a function of the  state space $\mathcal{S}$, for all
    rational Platonic, Archimedean, and Catalan solids.  The
    columns represent,  from left to right:  the geometrical
    characterization of  the state space  $\mathcal{S}$; the
    number  $m$ of  extremal  states;  the affine  dimension
    $\operatorname{aff.dim}   (   \mathcal{S}  )$;   whether
    $\mathcal{S}$ is  centrally symmetric  (CS) or  not (see
    Lemma~\ref{lmm:bounds}); the order  $| \mathcal{G} |$ of
    the symmetry  group $\mathcal{G}$ of  $\mathcal{S}$ and,
    under  the no-restriction  hypothesis, of  $\mathcal{E}$
    too;    the   number    $|\mathcal{M}|$   of    extremal
    measurements; the number $|\mathcal{M}'|$ of equivalence
    classes  of  extremal  measurements  up  to  symmetries;
    finally, the  exact value of the  signaling dimension of
    $\mathcal{S}$.}
  \label{tab:regular}
\end{table}

One question  that Table~\ref{tab:regular}  helps addressing
is: what  systems are  indistinguishable from  the classical
bit and  the qubit  in terms  of their  signaling dimension?
Notice       that,      as       a      consequence       of
Corollary~\ref{cor:sigdim2d},   the   only   two-dimensional
systems systems indistinguishable from the classical bit and
the qubit  in terms of  their signaling dimension  are those
whose     state     space    is     centrally     symmetric.
Table~\ref{tab:regular}   shows   that,   among   the
three-dimensional   systems   considered,  the   only   ones
indistinguishable from a classical bit  and a qubit in terms
of   signaling  dimension   are   the   cube,  the   rhombic
dodecahedron, and the truncated octahedron.

Other   sets  that   have   been   considered  include   the
hyper-octahedra    in   dimensions    up   to    five   (see
Table~\ref{tab:octa}).  Incidentally,  it is  perhaps  worth
noticing  that,  as  a  consequence of  the  fact  that  any
hyper-octahedral  set of  effects  admits only  two-outcomes
extremal measurements, the signaling dimension of any theory
with hyper-cubical  set $\mathcal{S}$ of states  is bound to
be two.

\begin{table}[h!]
  \begin{center}
    \begin{tabular}{|c|c|c|c|c|c|c|c|}
      \hline
      \textbf{$\mathcal{S}$} & $m$ & $\operatorname{aff.dim} ( \mathcal{S} )$ & CS & $|\mathcal{G}|$ & $|\mathcal{M}|$ & $|\mathcal{M}'|$ & $\operatorname{sig.dim} ( \mathcal{S} )$\\
      \hline
      Octahedron & 6          & 3       & True & 48           & 6      & 2     & 3      \\
      Hyper-octahedron & 8          & 4       & True & 384          & 48     & 3     & 3      \\
      Hyper-octahedron & 10         & 5       & True & 3840         & 2712   & 9     & 3 \\   
      \hline
    \end{tabular}
  \end{center}
  \tcaption{Exact value of the  signaling dimension for GPTs
    as  a function  of  the state  space $\mathcal{S}$,  for
    hyper-octahedra  up  to  dimension  five.   The  columns
    represent,   from  left   to   right:  the   geometrical
    characterization of  the state space  $\mathcal{S}$; the
    number  $m$ of  extremal  states;  the affine  dimension
    $\operatorname{aff.dim}   (   \mathcal{S}  )$;   whether
    $\mathcal{S}$ is  centrally symmetric  (CS) or  not (see
    Lemma~\ref{lmm:bounds});  the  order $|\mathcal{G}|$  of
    the symmetry  group $\mathcal{G}$ of  $\mathcal{S}$ and,
    under  the no-restriction  hypothesis, of  $\mathcal{E}$
    too;    the   number    $|\mathcal{M}|$   of    extremal
    measurements;    the    number    $|\mathcal{M}'|$    of
    representatives  of  equivalence   classes  of  extremal
    measurements up to symmetries;  finally, the exact value
    of the signaling dimension of $\mathcal{S}$.}
  \label{tab:octa}
\end{table}

Given that  the signaling  dimension of  the two-dimensional
octahedron   (the  square)   is   two,  and   that  of   the
three-dimensional octahedron  is three,  it might  have been
expected the signaling dimension  of the hyper-octahedron to
grow  with  the  affine   dimension.   This  expectation  is
motivated  by  the fact  that  the  dual  set, that  is  the
hyper-cube,  contains non-trivial  extremal measurements  in
any dimension  (for instance, it is  known~\cite{Mar11} that
the  regular  simplex can  be  inscribed  in the  hyper-cube
whenever Hadamard matrices exists,  e.g. whenever the linear
dimension  is a  multiple of  four up  to $664$).   However,
Table~\ref{tab:octa} disproves such an expectation, at least
up to dimension five.

Let  us  conclude  by  discussing  the  performance  of  our
algorithm  in  comparison   with  previous  proposals.   For
instance,    for    $d    =    3$   the    last    row    of
Table~\ref{tab:regular}  corresponds to  $m =  14$ and  $n =
24$, while the last  row of Table~\ref{tab:octa} corresponds
to $m  = 10$  and $n =  32$; as a  comparison, for  the same
value  of the  dimension $d  =  3$, by  using the  algorithm
proposed  in Ref.~\cite{DC21},  the Authors  thereof reached
maximum values of  $m = 4$ and $n =  12$. Our algorithm also
outperforms (by roughly  a factor of four)  the one proposed
and implemented  in Ref.~\cite{DTB24} in the  calculation of
the  signaling dimension  of the  composition of  two square
systems.

\section{Conclusion and outlook}
\label{sect:conclusion}

In  this work  we explored  the landscape  of the  signaling
dimension of  generalized probabilistic theories.   Our main
result  consists in  deriving necessary  conditions for  the
saturation  of  upper  and  lower bounds  on  the  signaling
dimension.  We discuss two  applications of such conditions.
First, we  show that  such conditions suffice  to completely
characterize  the signaling  dimension of  any system  whose
affine dimension is two, thus  proving that in that case the
signaling dimension  is completely determined by  whether or
not the system  is centrally symmetric. Second,  we show how
such conditions can be used to improve upon previously known
algorithms  for  the  exact  computation  of  the  signaling
dimension of any  given polytopic system, and  we apply such
improved  algorithms  to  compute  the exact  value  of  the
signaling   dimension  of   certain   classes  of   rational
polytopes.

We conclude  by discussing some  open problems. It  would of
course be  of the  utmost interest  to derive  a closed-form
solution to the signaling dimension problem (an optimization
problem) in terms  of the geometrical properties  of the set
of admissible states  in arbitrary dimension, as  we did for
the two-dimensional case.

Another issue consists of the fact that, as explained in the
previous   section,  for   merely   technical  reasons   the
implementation of  our algorithm  is restricted  to rational
polytopes.  This,  for instance,  excludes from  our results
certain   Platonic   solids   (the   icosahedron   and   the
dodecahedron), as  well as  certain Archimedean  and Catalan
solids.  Generalizing  the implementation and  computing the
signaling dimension for such solids would shed further light
on the problem of characterizing those state spaces that are
indistinguishable from  the classical  bit and the  qubit in
terms of their signaling dimension.

Another interesting  open question is whether  the signaling
dimension  of  regular  hyper-octahedra increases  with  the
affine  dimension  of the  system  (and,  if so,  how).   As
justified  in   the  previous  section,  this   question  is
particularly  interesting in  those  (affine) dimensions  in
which there  exist Hadamard  matrices, the smaller  of which
(after   three)  is   seven.   However,   the  computational
complexity of the  problem, along with the fact  that we are
applying  a  general-purpose  algorithm not  specialized  to
hyper-octahedra,  restricted our  analysis  up to  dimension
five; hence, the problem remains open.

\section*{Acknowledgments}

M.~D.  acknowledges support from  the Department of Computer
Science and Engineering, Toyohashi University of Technology,
from the International Research Unit of Quantum Information,
Kyoto  University, and  from the  JSPS KAKENHI  grant number
JP20K03774.

\end{document}